\documentclass[conference]{IEEEtran}
\IEEEoverridecommandlockouts
\usepackage{cite}
\usepackage{amsmath,amssymb,amsfonts}
\usepackage{algorithmic}
\usepackage{graphicx}
\usepackage{textcomp}
\usepackage{xcolor}

\usepackage{bbm}
\usepackage{setspace}
\usepackage{multirow}
\usepackage{cases}
\usepackage{subfigure}
\usepackage{hyperref}
\renewcommand{\baselinestretch}{0.98} 
\makeatletter
\IEEEtriggercmd{\reset@font\normalfont\fontsize{5.5pt}{8.2pt}\selectfont}
\makeatother
\IEEEtriggeratref{1}

\usepackage{booktabs}
\usepackage[T1]{fontenc}
\usepackage{comment}
\usepackage{multicol}
\usepackage[flushleft]{threeparttable}
\usepackage{amsthm}
\usepackage{algorithmic,algorithm}

\def\BibTeX{{\rm B\kern-.05em{\sc i\kern-.025em b}\kern-.08em
    T\kern-.1667em\lower.7ex\hbox{E}\kern-.125emX}}

\definecolor{orange}{rgb}{1, .36, .08}

\newtheorem{theorem}{\textbf{Theorem}}

\usepackage[all]{background}
\usepackage{stackengine}
\setstackEOL{\\}
\setstackgap{L}{\normalbaselineskip}
\SetBgContents{\color{blue}{\tiny \Longstack{PREPRINT - accepted at Design, Automation and Test in Europe Conference (DATE), 2022.}}}
\SetBgPosition{4.5cm,1cm}
\SetBgOpacity{1.0}
\SetBgAngle{0}
\SetBgScale{1.8}

\begin{document}

\title{
{TriLock: 
IC Protection with Tunable Corruptibility and Resilience to SAT and Removal Attacks}
\thanks{$^1$Yuke Zhang and Yinghua Hu contributed equally to this work. 
This work is based on research sponsored by the U.S. Government. The views and conclusions contained herein are those of the authors and should not be interpreted as necessarily representing the official policies or endorsements of the U.S. Government.}
}

\author{
\\[-3ex]
Yuke Zhang$^1$, Yinghua Hu$^1$, Pierluigi Nuzzo, and Peter A. Beerel\\ 
\small Department of Electrical and Computer Engineering, University of Southern California, Los Angeles, CA, USA \\ \{yukezhan, yinghuah, nuzzo, pabeerel\}@usc.edu\\[-2ex]
}

\maketitle

\begin{abstract}
Sequential logic locking has been studied over the last decade as a method to protect sequential circuits from 
reverse engineering. However, most of the existing sequential logic locking techniques are threatened by increasingly more sophisticated SAT-based attacks, efficiently using input queries to a SAT solver to rule out incorrect keys, as well as removal attacks based on structural analysis. In this paper, we propose $\mathtt{TriLock}$, a sequential logic locking method that simultaneously addresses these vulnerabilities. $\mathtt{TriLock}$ can achieve high, tunable functional corruptibility while still guaranteeing  exponential queries to the SAT solver in a SAT-based attack. Further, it adopts a state re-encoding method to 
obscure the boundary between the original state registers and those inserted by the locking method, thus making it more difficult to detect and remove the locking-related components.
\end{abstract}


\section{Introduction}\label{sec:intro}

The decentralization of the integrated circuit (IC) supply chain over the past few decades has increasingly raised concerns about potential threats, such as intellectual property (IP) piracy and hardware Trojan insertion~\cite{tehranipoor2010survey}.
One of the most investigated IC protection schemes against these threats is logic encryption (or locking)~\cite{chakraborty2009harpoon,rajendran2013fault,yasin2016sarlock,xie2018anti,hu2020sanscrypt_extended,pilato2021assure,chowdhury2021enhancing}, which 
adds programmability to the design at the gate or register-transfer level (RTL), so that the intended function is hidden from unauthorized users, and can only be accessed by a legal user by appropriately configuring the locked circuit. 

Early logic locking methods have mostly focused on modifying the combinational portion of a circuit~\cite{rajendran2013fault,yasin2016sarlock}. 
On the other hand, \emph{sequential logic locking}, the focus of this paper, usually involves creating new states in the finite state machine representing the original circuit and modifying its transitions~\cite{chakraborty2009harpoon,desai2013interlocking, dofe2018novel,kasarabada2019deep,roshanisefat2020dfssd,hu2020sanscrypt_extended,rezaei2021sequential}. 
%
The correct functionality is typically retrieved by either providing a key sequence, i.e., a dynamic sequence of key patterns, via the primary input ports during the first few clock cycles~\cite{chakraborty2009harpoon,desai2013interlocking} or by setting a set of key ports to fixed values throughout the circuit operation time~\cite{kasarabada2019deep,roshanisefat2020dfssd,rezaei2021sequential}. Sequential locking shows the promise of significantly increasing the attack effort at reasonable cost by judiciously expanding a circuit's state space. 
Yet, major threats to existing schemes have been posed by increasingly more sophisticated SAT-based attacks, efficiently using queries to a Boolean satisfiability (SAT) solver to rule  out incorrect keys, as  well as removal attacks that can exploit structural circuit signatures. 

SAT-based attacks~\cite{el2017reverse,shamsi2019kc2,hu2020sanscrypt_extended}, leveraging circuit unrolling and model checking, have shown to be successful against the first versions of sequential locking~\cite{chakraborty2009harpoon}, in that they can effectively exploit the early occurrence of output errors to dramatically decrease the number of SAT queries and accelerate the key search. This vulnerability has called for methods that can  intentionally postpone the first occurrence of the output errors~\cite{desai2013interlocking,dofe2018novel,kasarabada2019deep,rezaei2021sequential}.
However, SAT-based attacks can still be accelerated by leveraging functional corruptibility to help estimate the required circuit unrolling depth and further reduce the number of SAT queries~\cite{hu2021fun}. Moreover, SAT-resilient methods tend to exhibit poor error rates, usually captured in terms of \emph{functional corruptibility}, hence lack enough protection -- a trade-off that is extensively documented in the context of combinational locking~\cite{hu2019models,hu2021risk}.
Finally, the net boundary between the locking-related components and the rest of the circuit makes them vulnerable to removal attacks~\cite{meade2017revisit}, possibly boosted by 
machine learning-assisted netlist analysis tools~\cite{meade2016gate,geist2020relic}. \emph{A robust sequential locking scheme that can offer quantifiable protection and resilience to SAT-based and removal attacks is still elusive.}

This paper proposes $\mathtt{TriLock}$, 
an IC protection scheme that leverages the temporal dimension of sequential locking to break the well-known trade-off between SAT-attack resilience and functional corruptibility of combinational locking and simultaneously address all of the above challenges. 
Our contributions include: 
\begin{itemize}
    \item A cost-effective logic locking method that can exponentially increase the number of SAT queries required for a successful SAT-based attack. 
    \item An error-generation mechanism that can  strategically increase the output error rate to achieve a desired functional corruptibility without compromising SAT-attack resilience. 
    \item A state re-encoding technique that can significantly blur the boundary between the original circuit and the logic added by the locking scheme.
\end{itemize}
To the best of our knowledge, $\mathtt{TriLock}$ is the first sequential locking technique that simultaneously tackles all the above security objectives. We demonstrate its effectiveness via  security analysis and empirical validation on  ISCAS'89~\cite{brglez1989combinational} and ITC'99\cite{corno2000rt} benchmarks. 


\section{Background and Related Work}\label{sec:background}

We discuss SAT-based and removal attacks as well as the methods that have been proposed to counteract them. 

\subsection{Preliminaries}
 
In the following, we just use the term circuit to refer to a sequential circuit. 
\begin{figure}[t]
    \centering{
    \subfigure[]{
    \includegraphics[width=0.21\columnwidth]{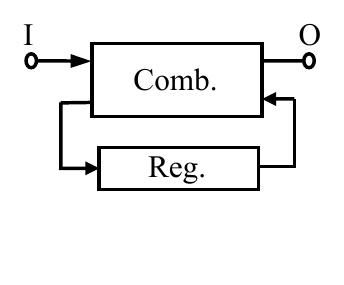}\label{fig:seq_circuit}}
    \subfigure[]{
    \includegraphics[width=0.72\columnwidth]{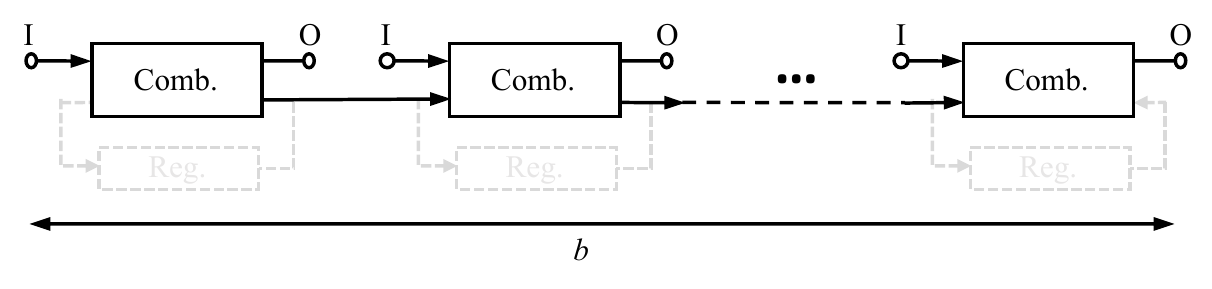}\label{fig:unrolled_circuit}}
    }
    \vspace{-3mm}
    \caption{Schematic of (a) a sequential circuit and (b) its $b$-unrolled version.}
    \label{fig:from_seq_to_comb}
    \vspace{-5mm}
\end{figure}
For a circuit $C$, shown in Fig.~\ref{fig:seq_circuit}, its $b$-unrolled version, denoted by $C_b$, is a combinational circuit, shown in Fig.~\ref{fig:unrolled_circuit}, that
represents the behavior of $C$ over the first $b$ clock cycles. We denote by $I$ and $O$ the sets of input and output ports of $C$, respectively, by $k$ and $i$ its key and input sequence, respectively, and by $\kappa$ the (cycle) length of $k$, i.e., the number of clock cycles required to provide $k$ to the circuit after reset.
For a sequence $s$, $s_n$ denotes the sub-sequence of $s$ associated with the $n$-th unrolling (clock cycle) and $s_{n\leftrightarrow m}$ the one associated with the range of unrollings from $n$ to $m$.  

Let $C^{o}_b$ be the $b$-unrolled version of the original circuit and $C^{e}_b$, with a slight abuse of notation, the $(\kappa+b)$-unrolled version of the encrypted circuit. For simplicity, we refer to $C^{e}_b$ as the  $b$-unrolled version of the encrypted circuit, by skipping the first $\kappa$ cycles used to input the key. 
Let the functions implemented by $C^{o}_b$ and $C^{e}_b$ be $f_b:\mathbb{B}^{b|I|}\rightarrow \mathbb{B}^{b|O|}$ and  $f_{b}':\mathbb{B}^{\kappa|I|}\times \mathbb{B}^{b|I|}\rightarrow \mathbb{B}^{b|O|}$, respectively. Then,  
the \emph{functional corruptibility} (FC) of a $b$-unrolled version of an encrypted circuit is defined as~\cite{hu2021fun}
\begin{equation}
    FC_b = \frac{1}{2^{(\kappa + b)|I|}} \sum_{i\in \mathbb{B}^{b|I|}} \sum_{k\in \mathbb{B}^{\kappa|I|}}\mathbbm{1}(f_b(i)\neq f_{b}'(i,k)),
    \label{eq:fc_def}
\end{equation}
\noindent where $\mathbbm{1}(\cdot)$ is the indicator function.  
$FC_b$ quantifies the proportion of errors over all input-key combinations for a $b$-unrolled encrypted circuit. 

\subsection{SAT-Based Attacks}
The idea of formulating a SAT problem to prune out wrong keys was first adopted by the SAT attack \textsc{Comb-SAT}~\cite{subramanyan2015evaluating} to combinational logic locking.  
\textsc{Comb-SAT} assumes that the attacker has access to the netlist of the locked circuit as well as unlimited access to the correct input/output pairs from the original circuit. 
An iterative key elimination process is executed by searching for distinguishing input patterns (DIPs) via SAT solving.
A DIP is an input pattern of the locked circuit for which there exist 
two different keys that lead to different outputs.
When a DIP $d$ is found, it can effectively rule out a set of wrong keys $K_d$ that are detectable by $d$, expressed as 
\begin{equation}
    K_d = \{k|f({d})\neq f'(d, k)\},
    \label{eq:key_pruned}
\end{equation}
where $f$ and $f'$ are the functions implemented by the original and the locked circuit, respectively.
Until the correct key is obtained, more DIPs are iteratively found and used to further prune out the key search space. 
A trade-off exists between SAT-attack resilience, i.e., the number of DIPs required to find the correct key, and the FC of a locked circuit~\cite{hu2019models,hu2021risk}. The larger the number of errors induced by a wrong key, the higher the likelihood that the wrong key can be detected and eliminated by a DIP. 

\textsc{Comb-SAT} cannot be directed applied to sequential circuits without scan access to their internal states. It can, however, be extended by relying on circuit unrolling to generate a $b$-unrolled version of the encrypted circuit, $C^{e}_b$, on which to apply \textsc{Comb-SAT}~\cite{el2017reverse,shamsi2019kc2,hu2020sanscrypt_extended}.
Once a key is found for $C^{e}_b$, model checking is performed to verify whether the key is also correct for $C^{e}$, beyond the first $b$ clock cycles. If this check fails, the above steps will be repeated for a larger $b$. Several sequential encryption methods~\cite{desai2013interlocking,dofe2018novel,kasarabada2019deep,rezaei2021sequential} boost SAT-attack resilience by increasing the minimum unrolling depth $b^*$ that is needed to rule out all the wrong keys. However, $b^*$ has been recently shown to be effectively predictable~\cite{hu2021fun}, thus making SAT-based attacks even more efficient. 


\subsection{Removal Attacks}

Sequential logic encryption methods may be vulnerable to  removal attacks based on structural analysis of the circuit netlist~\cite{meade2016gate,meade2017revisit,geist2020relic,roshanisefat2020dfssd}. 
Unwanted signatures may be detected in the state transition graph (STG) of the encrypted circuit, for example, when there is only one edge from the set of states added by the locking logic to the states in the original STG~\cite{chakraborty2009harpoon}.
Graph analysis methods can then be applied to the STG to recognize the boundary between the two sets of states~\cite{meade2017revisit}. State-Deflection~\cite{dofe2018novel} adds several sink state clusters in the STG to trap illegal users. 
However, because a sink cluster does not have any outgoing edge, it can be easily identified by a strongly connected component (SCC) algorithm. 

Several papers~\cite{meade2016gate,geist2020relic} view the recognition of state registers as the first step for reverse-engineering finite state machines, and propose accurate tools for this task. 
After the state registers are recognized, the original registers must be separated from the additional registers associated with the encryption logic. In this paper, we assume that all the state registers of a circuit can be successfully recognized. 
The aim of $\mathtt{TriLock}$ is to make it more difficult to separate and remove the additional registers associated with the encryption.

\section{TriLock}\label{sec:method}


\begin{figure}[t]
    \centering
    \includegraphics[width=\columnwidth]{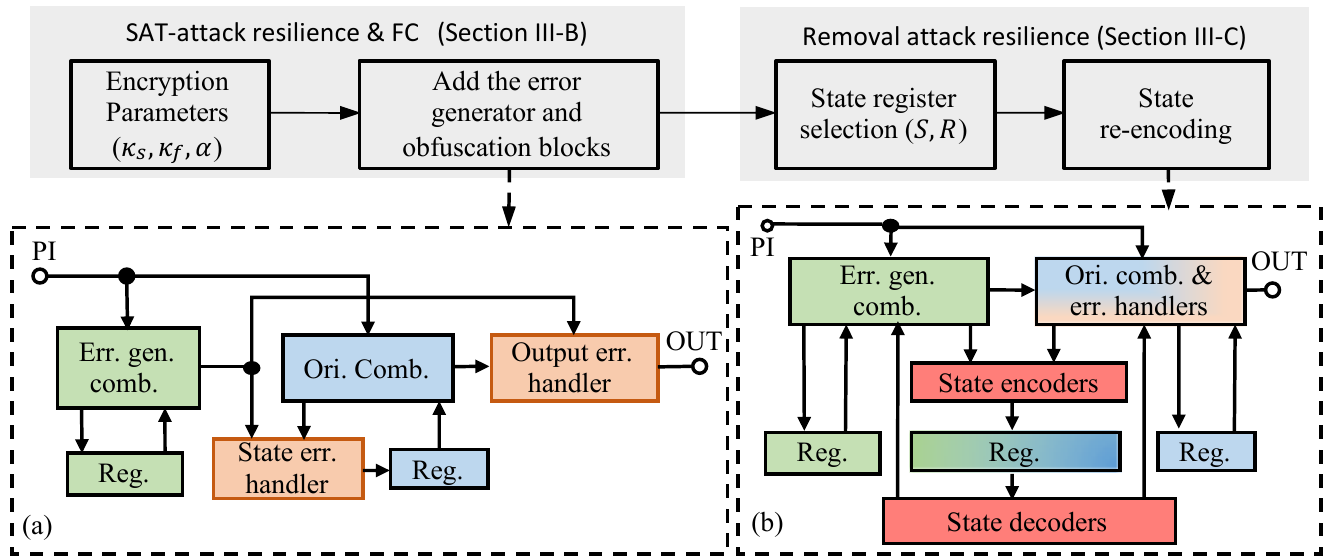}
    \vspace{-6mm}
    \caption{Overview of $\mathtt{TriLock}$.}
    \vspace{-5mm}
    \label{fig:overview}
\end{figure}

We first discuss the trade-off between SAT-attack resilience and FC and present in  Section~\ref{subsec:SAT_resilience} a naive implementation of $\mathtt{TriLock}$ that achieves high resilience at the cost of low FC. We show how to overcome the trade-off 
in Section~\ref{subsec:FC}, 
allowing independent configuration of FC without compromising SAT-attack resilience. We finally detail our 
strategy to mitigate removal attacks 
in Section~\ref{subsec:state_reencoding}. 
The encryption flow of $\mathtt{TriLock}$ 
is shown in Fig.~\ref{fig:overview}. 

\subsection{Trade-off Between SAT-Attack Resilience and FC}
\label{subsec:SAT_resilience}

Combinational logic encryption techniques, such as  SARLock~\cite{yasin2016sarlock} and Anti-SAT~\cite{xie2018anti}, adopt point functions to achieve exponential SAT-attack resilience, quantified in terms of the required number of DIPs ($n_{dip}$). 
In these methods, each DIP can only rule out a limited number of wrong keys at each iteration of the SAT attack.  
We apply a similar concept in $\mathtt{TriLock}$ to achieve exponential $n_{dip}$ in the key length $\kappa$. 
We assume that an attacker can efficiently estimate the minimum required unrolling depth $b^*$~\cite{hu2021fun} and   
perform \textsc{Comb-SAT} directly on  
$C^e_{b}$, with $b \geq b^*$. We then focus on guaranteeing an exponential $n_{dip}$ for $C^e_{b}$. 

We define an \emph{error function}, $E_b: \mathbb{B}^{b|I|} \times \mathbb{B}^{\kappa |I|} \rightarrow \{1,0\}$, for   $C^e_{b}$, 
as a function that takes as arguments a $b|I|$-bit input sequence $i$ and a $\kappa|I|$-bit key sequence $k$ and returns $1$ if and only if an error occurs at the output of $C^e_{b}$, i.e., if and only if $f_{b}'(k, i) \neq f_{b}(i)$ holds. 
A naive error function that achieves exponential $n_{dip}$ 
can then be obtained by setting $b^*=\kappa$ and by implementing a point function, as done, for example, in SARLock~\cite{yasin2016sarlock}. 
We would therefore obtain 
\begin{equation}\label{eq:sat}
  E_{b}^{N}(i, k) = \mathbbm{1} \left[ (k \neq k^* ) \wedge (k=i_{1\leftrightarrow\kappa}) \right],
\end{equation}
where $k^*$, the correct key sequence, is a fixed sequence of length $\kappa$. 
For an arbitrary wrong key $k^w$, there exists a set of input sequences $IS_{k^w}$ for which $E_{b}^{N}$ evaluates to $1$, 
expressed as follows,
\begin{equation}
\label{eq:IS_set}
    IS_{k^w}=\{i\in\mathbb{B}^{b|I|}|k^w=i_{1\leftrightarrow\kappa}\},
\end{equation}
that is, all the input sequences having $k^w$ as a prefix.
Based on the mechanism of \textsc{Comb-SAT}, any input sequence in $IS_{k^w}$ can then be selected as a DIP to rule out the wrong key $k^w$. However, an input sequence in $IS_{k^w}$ cannot detect any other wrong key, that is, 
\begin{equation}
    \forall i\in IS_{k^w},\forall k\in\mathbb{B}^{\kappa |I|}\backslash \{k^w\}, E_b(i,k)=0.
\end{equation}
Consequently, one DIP can only rule out one wrong key at a time and 
the $n_{dip}$ will equal the number of wrong keys, i.e., 
\begin{equation} 
    n_{dip}=2^{\kappa|I|}-1.
    \label{eq:est_dips}
\end{equation}
\begin{figure}[t]
    \centering{
    \vspace{-2mm}
    \subfigure[]{
    \includegraphics[width=0.22\columnwidth]{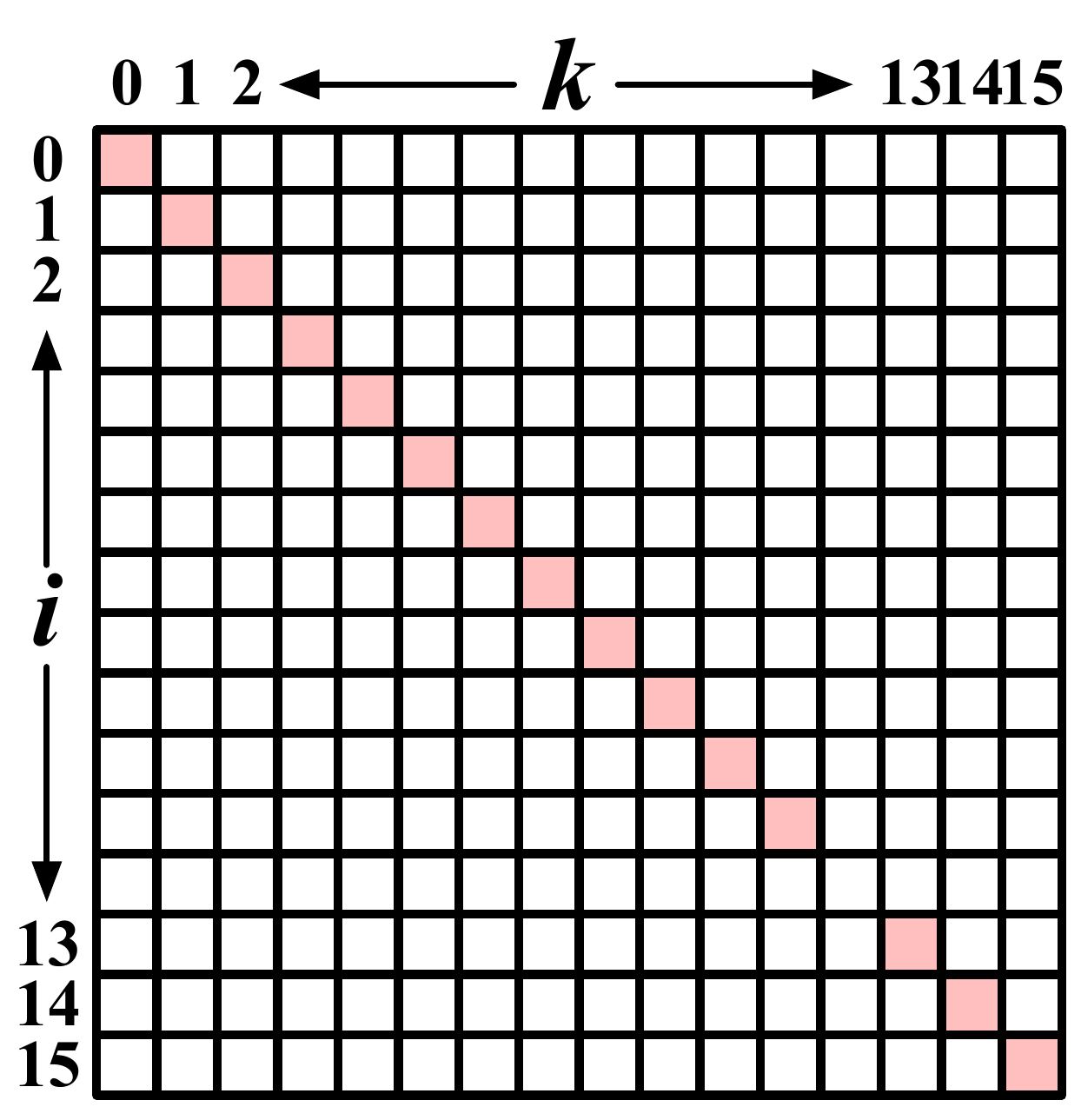}\label{fig:sarlock_table}}
    \subfigure[]{
    \includegraphics[width=0.73\columnwidth]{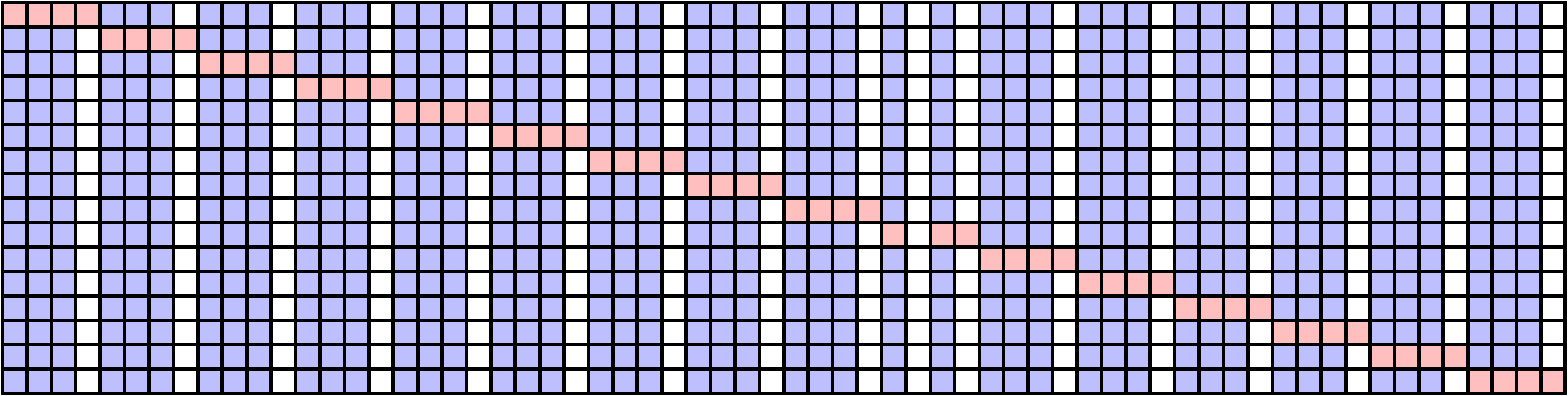}\label{fig:trilock_table}}
    }
    \vspace{-5mm}
    \caption{Error tables resulting from the application of (a) $E_{b}^{N}$ as in~\eqref{eq:sat}, with $|I|=\kappa =b^*=b=2$ and (b) $E_{b}^{SF}$ as  in~\eqref{eq:trilock}, with $|I|=\kappa_s =b^*=b=2$ and $\kappa_f=1$.}
    \label{fig:error_table}
    \vspace{-5mm}
\end{figure}
The effect of $E_{b}^{N}$ is pictorially represented by the colored error table in Fig.~\ref{fig:sarlock_table} for a 2-input circuit with $\kappa =b^*=b=2$. 
The row and the column indexes correspond to the values of the input and the key sequences, respectively. 
If $E_{b}^{N}(i,k)=1$, the corresponding square is red. 
By definition, $FC_{b}$ can be computed as 
\begin{equation}
    FC_{b}=\frac{(2^{\kappa|I|}-1)\cdot 2^{(b-\kappa)|I|}}{2^{(\kappa + b)|I|}} \approx \frac{1}{2^{\kappa|I|}} = \frac{1}{n_{dip}+1}.
    \label{eq:fc_naive}
\end{equation}
which is, unsurprisingly, low for $E_{b}^{N}$, as low as $0.06$ in the scenario of Fig.~\ref{fig:sarlock_table}.
\begin{figure}[t]
    \centering{
    \vspace{-2mm}
    \subfigure[]{
    \includegraphics[width=0.47\columnwidth]{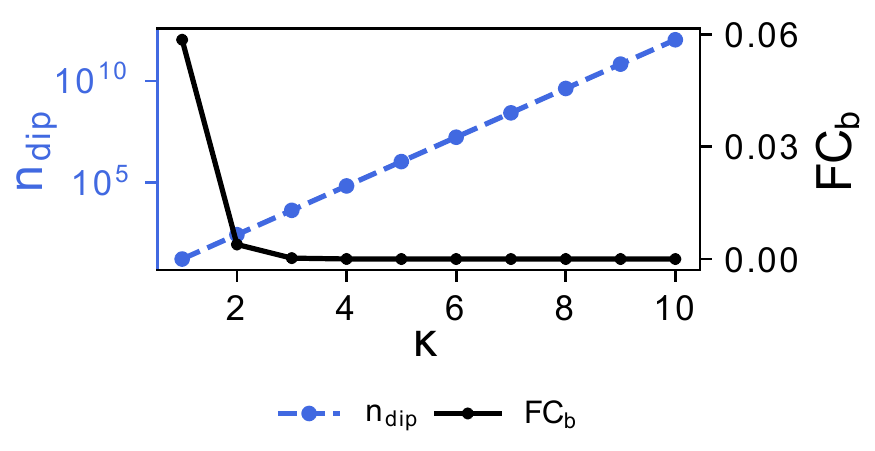}\label{fig:dips_fc_1}}
    \subfigure[]{
    \includegraphics[width=0.47\columnwidth]{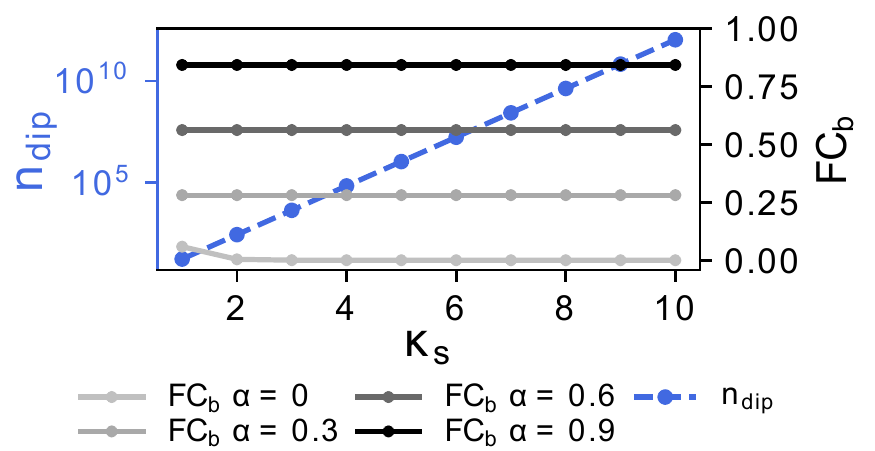}\label{fig:dips_fc_2}}
    }
    \vspace{-5mm}
    \caption{Relations between $n_{dip}$ and $FC_b$ of a 4-input circuit implemented by (a) $E_b^N$, and (b) $E_b^{SF}$ with $\kappa_f=1$.}
    \label{fig:dips_fc}
    \vspace{-5mm}
\end{figure}
In the following, we detail how this trade-off between $n_{dip}$ and $FC_b$, captured by~\eqref{eq:fc_naive} and pictorially shown in Fig.~\ref{fig:dips_fc_1} for different key cycle lengths $\kappa$ for a 4-input circuit, can be circumvented by appropriately designing the error function.

\subsection{Circumventing the SAT-Attack Resilience vs. FC Trade-Off}
\label{subsec:FC}

For better clarity, we use $\kappa_s$ to denote the key cycle length and rewrite the error function in~\eqref{eq:sat} as follows:
\begin{equation}\label{eq:sat_trilock}
  E_{b}^{S}(i, k) = \mathbbm{1} \left[ (k \neq k^* ) \wedge (k_{1\leftrightarrow\kappa_s}=i_{1\leftrightarrow\kappa_s}) \right].
\end{equation}
We can now increase FC without compromising the attack resilience achieved by $E_{b}^{S}$ 
by strategically redesigning the error function over 
a larger key cycle length $\kappa=\kappa_s + \kappa_f$, leading to an extended error table, as shown in 
Fig.~\ref{fig:trilock_table} for a 2-input circuit with $\kappa_s=b^*=b=2$ and $\kappa_f=1$. 
The red squares represent the errors defined by $E_b^S$; we denote their number by $n_b^{S}$, given by  
\begin{equation}
    \label{eq:nDIP}
    n_b^{S}=(2^{\kappa|I|}-1)\cdot 2^{(b-\kappa_s)|I|}\approx 2^{(b+\kappa_f)|I|}.
\end{equation}
Similarly to \eqref{eq:IS_set}, for any wrong key with prefix $k^w_{1\leftrightarrow\kappa_s}$, there exists a set of input sequences  $IS_{k^w}=\{i\in\mathbb{B}^{b|I|}|k^w_{1\leftrightarrow\kappa_s}=i_{1\leftrightarrow\kappa_s}\}$ that can be used as DIPs to eliminate only wrong keys with the same prefix $k^w_{1\leftrightarrow\kappa_s}$.
There are, in total, $2^{\kappa_s|I|}$ possible values for the prefix $k^w_{1\leftrightarrow\kappa_s}$ of a wrong key, so the SAT-attack resilience corresponding to $E_b^S$ is 
\begin{equation}\label{eq:sat_resilience_trilock}
    n_{dip} = 2^{\kappa_s|I|}.
\end{equation}
Besides the errors defined by $E_b^S$, we look for a set of additional input-key pairs in $C_b^e$ such that, if an error is added at each pair, it will not decrease the SAT-attack resilience achieved by $E_{b}^{S}$. For a fixed sequence $k^{**}$ of length $\kappa_f$, specified by the designer and such that $k^{**}\neq k^*_{(\kappa-\kappa_f)\leftrightarrow\kappa}$, one such set can be defined as follows: 
\begin{equation}
    \label{eq:error_free_set}
    P_b^{k^*, k^{**}} = \{(i,k)|k_{(\kappa-\kappa_f)\leftrightarrow\kappa}\neq k^{**}) \wedge(k\neq k^{*})\},
\end{equation}
where $k^*$ is the correct key sequence of length $\kappa$.
The blue squares in Fig.~\ref{fig:trilock_table} pictorially represent $P_b^{k^*, k^{**}}$ when $k^{*}=100101$ and $k^{**}=11$.
The following result states the property of $P_b^{k^*, k^{**}}$. 

\begin{theorem}\label{theorem:DIP_unchanged}
Given two sequences $k^{**}$ and $k^*$ of length $\kappa_f$ and $\kappa$, respectively, with $\kappa_f < \kappa$, let the $b$-unrolled version of the encrypted circuit $C_b^e$ implement the error function $E_b^S$ in~\eqref{eq:sat_trilock}, with $\kappa_s = \kappa -\kappa_f$. We assume that errors are inserted at all input-key pairs in $P_b^{k^*, k^{**}}$, defined as in~\eqref{eq:error_free_set}, i.e., $f_b'(i,k)\neq f_b(i),\ \forall (i,k)\in P_b^{k^*, k^{**}}$. Then, \textsc{Comb-SAT} on $C_b^e$ will require at least 
$2^{\kappa_s|I|}$ DIPs.
\end{theorem}
\begin{proof}
Given a wrong key $k^w$ with $k^{**}$ as a suffix, i.e., such that 
    $(k^w\neq k^*)\wedge (k^w_{(\kappa-\kappa_f)\leftrightarrow\kappa} = k^{**})$,
by the definition of $P_b^{k^*, k^{**}}$, we have that $(i,k^w) \notin P_b,\ \forall i\in \mathbb{B}^{b|I|}$.
Let $i^w$ be the DIP capable of detecting $k^w$. Since $i^w$ cannot be in $P_b^{k^*, k^{**}}$, then it must satisfy $E_b^S({i^w}, k^w)=1$. 
Let us now assume that $k^{v}$ is another wrong key with $k^{**}$ as a suffix, and such that $k^{w} \neq k^{v}$. By~\eqref{eq:sat_trilock}, we conclude that $E_b^S(i^w, k^{v})=0$ holds. Therefore, the DIP that allows ruling out $k^{w}$ cannot exclude any other wrong key having $k^{**}$ as a suffix. 
In total, there are $2^{\kappa_s|I|}$ wrong keys with a suffix of $k^{**}$. Therefore, 
$2^{\kappa_s|I|}$ DIPs are at least required to exclude those wrong keys. 
\end{proof}

Theorem~\ref{theorem:DIP_unchanged} indicates that more errors can be added to $C_b^e$ to boost $FC_b$ without negatively affecting the SAT-attack resilience achieved with $E_b^S$ alone. Moreover, the number of DIPs is independent of $b$. 
We denote by $n_b^{ef}$ the number of error-free entries on the error table.
We can compute the maximum achievable $FC_b$ as follows: 
\begin{equation}
\small
    \begin{aligned}
    \label{eq:trilock_max_fc}
    FC_b &= 
    \frac{2^{(\kappa+b)|I|} - n_b^{ef}}{2^{(\kappa+b)|I|}} 
    = 1 \!-\! \frac{2^{\kappa_s|I|}\cdot 2^{b|I|}}{2^{(\kappa+b)|I|}}
    =1\!-\!\frac{1}{2^{\kappa_f|I|}}.
    \end{aligned}
\end{equation}
In the scenario of Fig.~\ref{fig:trilock_table}, if all the blue squares are selected as errors, $FC_b$ can be as high as $0.75$. We can select the additional errors via the following error function 
\begin{equation}
    E_b^F(i,k)=\mathbbm{1} \left[
   (i,k)\in P_b^{k^*,k^{**}}
    \wedge r(i,k)
    \right],
\end{equation}
where $r(i,k)$ modulates the proportion of 
input-key pairs in $P_b^{k^*,k^{**}}$ that are selected to place an error. In this paper, we choose 
\begin{equation}
    r(i,k)=\mathbbm{1}\left[ k_{(\kappa-\kappa_f)\leftrightarrow\kappa} \leq \alpha(2^{\kappa_f|I|}-1) \right],
\end{equation}
where $\alpha \in (0,1)$ is a design parameter used to configure the desired FC to the following value: 
\begin{equation}
    \label{eq:trilock_max_fc_tune}
    FC_b \approx \alpha\left(1-\frac{1}{2^{\kappa_f|I|}}\right).
\end{equation}
%
%
%
%
By combining $E_b^S$ and $E_b^F$, we obtain
\begin{equation}\label{eq:trilock}
    \begin{aligned}
    E_b^{SF}(i,k) =
    E_b^{S}(i,k)\vee E_b^{F}(i,k),
    \end{aligned}
\end{equation}
which is the error function adopted by $\mathtt{TriLock}$ to guarantee exponential SAT-attack resilience and independently configurable FC. As shown in Fig.~\ref{fig:dips_fc_2} for a $4$-input circuit with $\kappa_f=1$, it is indeed possible to independently tune the $FC_b$ while still keeping high SAT-attack resilience. 
Moreover, the $n_{dip}$ and $FC_b$ in~\eqref{eq:sat_resilience_trilock} and~\eqref{eq:trilock_max_fc_tune}, respectively, are independent of the unrolling depth $b$.

We implement the error function $E_b^{SF}$ with the error generator block, shown in green in Fig.~\ref{fig:overview}(a), whose output signals are passed to the state error handler and the output error handler in orange, to trigger a signal inversion on a configurable number of state registers and primary output ports, respectively.

\subsection{Enhancing Removal Attack Resilience: State Re-encoding}\label{subsec:state_reencoding}

As shown in Fig.~\ref{fig:overview}(a), we can distinguish  
the original state registers of an an encrypted circuit (in blue) from the extra state registers added by the encryption (in green). Identifying the type of registers
is an essential step toward removal attacks.  
Specifically, an attacker can leverage the SCC algorithm
in a register connection graph (RCG), where a register is represented by a node and the existence of a path between two registers is denoted by a directed edge between the  corresponding nodes. 
The output of the SCC algorithm on an RCG is one or more clusters of nodes, called SCCs. For any two nodes in the same SCC, they are reachable from each other.
We denote by O-SCC, E-SCC, and M-SCC, an SCC containing only the original registers, only the extra registers, and a mix of the two types of registers, respectively. 

When no SCC in an RCG is an M-SCC, 
the identification of the set of original or extra registers is expected to be easy,
as each SCC is already a congregation of either original or extra registers. In the best case, an attacker could expect only two SCCs, an O-SCC and an E-SCC, with all the original registers and all the extra registers, respectively, as the algorithm output.  
Such a successful clustering of the registers would be due to the insufficient 
connections between original and extra registers. In contrast, if the connections are dense, one or more M-SCCs will exist and it will be harder to classify the type of registers in those M-SCCs. 
We then propose a \emph{state re-encoding} method, 
implemented on the encrypted sequential circuit, that intentionally creates new edges between O-SCCs and E-SCCs,
resulting in more registers 
being clustered in one or several 
M-SCCs. As shown in Fig.~\ref{fig:overview}(b), the state re-encoding method selects a configurable number of registers, and inserts state encoders and decoders after adding the error generator and error handlers.
We introduce below the register selection procedure and the encoder/decoder mechanism adopted in state re-encoding. 

\textbf{State Register Selection.} 
\begin{algorithm}[t]
 \caption{\texttt{State register selection}}
 \label{algo:state_register_selection}
 \footnotesize
 \begin{algorithmic}[1]\label{alg:attack_main}
 \renewcommand{\algorithmicrequire}{\textbf{Input:}}
 \renewcommand{\algorithmicensure}{\textbf{Output:}}
\REQUIRE $C^{e}$ and $S$. 
\ENSURE  $R$.
\STATE $RCG={\rm create\_graph}(C^e); \ R=[\ ]; count=0$
\STATE $E, O, M=run\_scc(RCG)$
\WHILE{\textbf{(($E \neq \emptyset$ or $O \neq \emptyset$) and $count < S$)}}
    \IF{($E \neq \emptyset$ and $O \neq \emptyset$)}
        \STATE $SCC_1, SCC_2 = get\_largest\_scc(E, O)$
    \ELSE
        \STATE $temp = get\_non\_empty\_set(E, O)$  
        \STATE $SCC_1, SCC_2 = get\_largest\_scc(temp, M)$
    \ENDIF

    \STATE $r_1, r_2 = get\_max\_edge\_node(SCC_1, SCC_2)$
     \STATE $R.append([r_1, r_2]); \ count =count  + 1$
    \STATE $RCG=update\_graph(RCG, [r_1, r_2])$
    \STATE $E, O, M=run\_scc(RCG)$
   
\ENDWHILE
 \RETURN $R$
 \end{algorithmic} 
 \end{algorithm}
We adopt a greedy method to iteratively select and encode pairs of original and extra registers.  Algorithm~\ref{algo:state_register_selection} shows the pair selection process,
which takes as inputs an encrypted netlist $C^e$ and the desired number of register pairs $S$, and returns a list of register pairs $R$ as output. 
After creating the RCG from $C^e$ (line 1) and running the SCC algorithm (line 2), three sets, i.e., $E$, $O$, and $M$, are generated that contain all the E-SCCs, O-SCCs, and M-SCCs, respectively.
To maximize the impact of state re-encoding for a single pair of original and extra registers, we first identify the largest O-SCC and E-SCC as $SCC_1$ and $SCC_2$, respectively (line 5). 
In case there does not exist an E-SCC or O-SCC, we choose the largest M-SCC as the replacement (line 7-8).
In $SCC_1$ and $SCC_2$, we then select the nodes connected by the largest number of edges, denoted as $r_1$ and $r_2$, respectively (line 10). 
We record $(r_1,r_2)$ in $R$ (line 11) and update the RCG (line 12-13) as a result of re-encoding $(r_1,r_2)$. 
The above register pair selection process is iterated until the designer-specified number of pairs is reached or no E-SCC and O-SCC exist. 


\textbf{Encoder/Decoder Mechanism.} 
For each pair of registers $(r_1,r_2)$ in $R$,
a state encoder and a state decoder are inserted between the combinational logic and the two registers, $r_1$ and $r_2$, as shown in Fig.~\ref{fig:st_structure},
to merge the SCC containing $r_1$ ($SCC_1$) and the one containing $r_2$ ($SCC_2$). 
We denote by $enc(\cdot)$ and $dec(\cdot)$ the functions of the encoder and the decoder, respectively. 
Between the encoder and the decoder, $r_1$ and $r_2$ are replaced by a set of encoded state registers $r_{ei}$, where $1 \leq i \leq m$. 
We denote by $s_1$ ($s_1'$) the net connecting from (to) the combinational logic to (from) $r_1$.
Similar notations are used for $r_2$.

To prevent the encoder/decoder structure from affecting the circuit function, a fixed-point condition $dec(enc(a))=a$ should hold for any 2-bit sequence $a$.
Moreover, state re-encoding should achieve successful merging of the two SCCs into one M-SCC, which requires the existence of the following looped signal propagation path (abbreviated as a path): 
\begin{equation} \label{eq:state_reg_loop}
\small
   \exists x,y\in[1,m]: s_1\rightarrow \mathbf{r_{ex}} \rightarrow s_2' \rightarrow s_2\rightarrow \mathbf{r_{ey}} \rightarrow s_1' \rightarrow s_1. 
\end{equation}
When such a looped path exists, any register $r_a$ in $SCC_1$ can connect to any register $r_b$ in $SCC_2$ via the path: $ r_a\rightarrow 
s_1 \rightarrow \mathbf{r_{ex}} \rightarrow s_2' 
\rightarrow r_b$, as shown in Fig.~\ref{fig:st_scc}. Similarly, $r_b$ can also reach $r_a$ via the path: $r_b 
\rightarrow s_2\rightarrow \mathbf{r_{ey}} \rightarrow s_1' 
\rightarrow r_a$.
On the RCG, the above paths construct a bidirectional edge between $r_a$ and $r_b$, which merges $SCC_1$ and $SCC_2$ into an M-SCC.

\begin{figure}[t]
    \centering{
    \vspace{-5mm}
    \subfigure[]{
    \includegraphics[width=0.55\columnwidth]{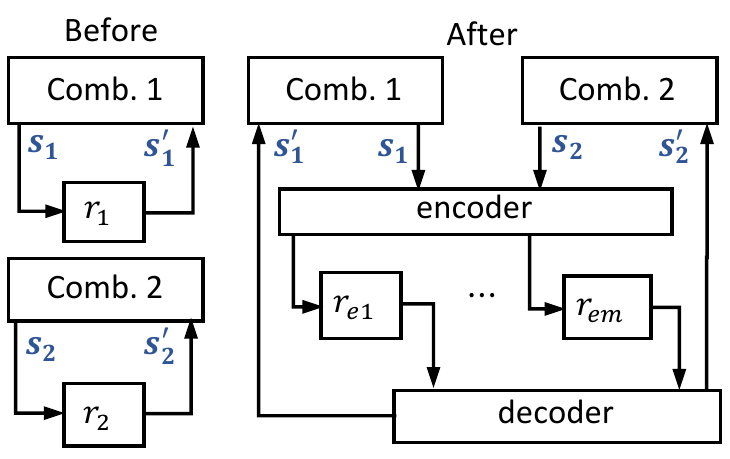}\label{fig:st_structure}}
    \subfigure[]{
    \includegraphics[width=0.36\columnwidth]{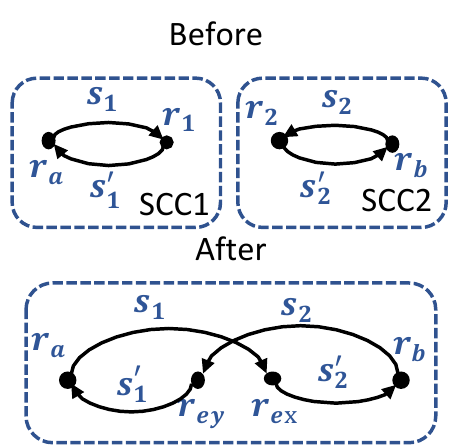}\label{fig:st_scc}}
    }
    \vspace{-3mm}
    \caption{(a) Abstract schematic and (b) RCG before and after state re-encoding.}
    \label{fig:re-encoding}
    \vspace{-5mm}
\end{figure}

In this paper, we implement the encoder with two arithmetic operations, namely, $e_1=s_1 + s_2$ and $e_2 = s_1 - s_2$, where $e_1$ and $e_2$ are the encoded states. 
The decoder also excutes two arithmetic operations, 
$s'_1 = \frac{1}{2}(e'_1+e'_2)$ and $s'_2 = \frac{1}{2}(e '_1-e'_2)$, which satisfies the fixed-point condition while creating a looped path as in \eqref{eq:state_reg_loop}. 
The associated two SCCs are, thus, merged into an M-SCC.  
To mitigate the structural signature produced by repeatedly implementing the same encoder/decoder, various $enc(\cdot)$ and $dec(\cdot)$ can be applied to different register pairs, which can be subject of  future work.

\section{Experimental Results}








We implement the encryption flow of $\mathtt{TriLock}$ in Python, using Synopsys Design Compiler and a 45nm Nangate Open Cell Library~\cite{nangate} as the synthesis tool and the target library, 
respectively. 
FC is simulated
with $800$ random inputs and keys using Synopsys VCS, while the 
SAT-attack resilience is evaluated via an implementation of 
a state-of-the-art SAT-based attack~\cite{hu2021fun} which can effectively predict the minimum required unrolling depth $b^*$. In the case of $\mathtt{TriLock}$, $b^*=\kappa_s$.
We select ten benchmark circuits from ISCAS'89~\cite{brglez1989combinational} and ITC'99~\cite{corno2000rt}, as shown in Table~\ref{table:SATresults}. All experiments are executed on an Intel(R) Xeon(R) E5-2450 2.5-GHz CPU with 126-GB memory. 

\begin{figure*}[t]
    \centering
    \includegraphics[width=\textwidth]{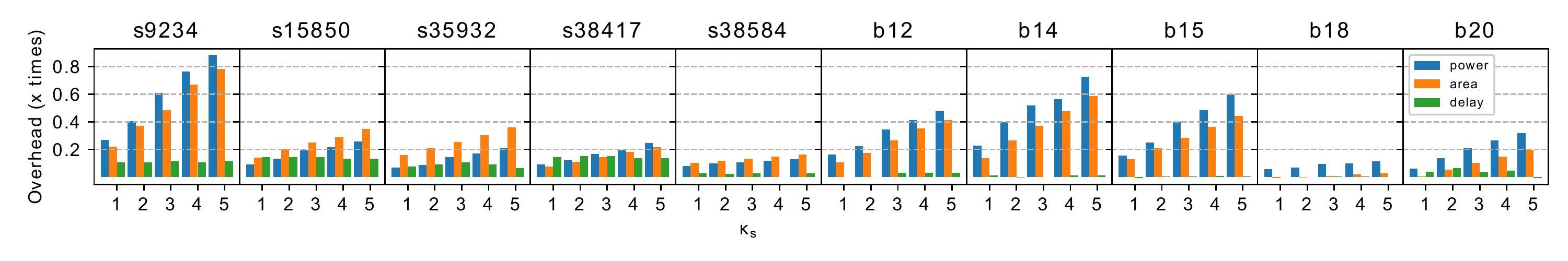}
    \vspace{-9mm}
    \caption{Area, power, and delay overhead of $\mathtt{TriLock}$ with $\kappa_f=1$, $\alpha = 0.6$, and $S=10$.}
    \label{fig:overhead}
    \vspace{-7mm}
\end{figure*}

\begin{table}[t]
\Large
\centering
\caption{SAT-attack resilience of TriLock}
\label{table:SATresults}
\setlength{\aboverulesep}{0pt}
\setlength{\belowrulesep}{0pt}
\resizebox{\columnwidth}{!}{

\small
\begin{tabular}{@{}c|cccc|cc|cc|cc}
\toprule
\multirow{2}{*}{Circuit} & \multicolumn{4}{c|}{Circuit Info.} & \multicolumn{2}{c|}{$\kappa_s = 1$} & \multicolumn{2}{c|}{$\kappa_s = 2$} & \multicolumn{2}{c}{$\kappa_s = 3$} \\ \cline{2-11}

 & PI & PO & FF & Gate & $n_{dip}$\tnote{1} & T (s) \tnote{1} & $n_{dip}$\tnote{1} & T (s) \tnote{1} & $n_{dip}$\tnote{1} & T (s) \tnote{1} \\ \hline
 
s9234 & 19 & 22 & 228 & 5597 & {\color{blue} 524288 }  & {\color{blue} 3.9e+06 } & {\color{blue} 2.7e+11 } & {\color{blue} 2.1e+12 } & {\color{blue} 1.4e+17 } & {\color{blue} 1.1e+18 } \\ \hline

s15850 & 13 & 87 & 597 & 9772 & 8192 & 105283  & {\color{blue} 6.7e+07 } & {\color{blue} 5.0e+08 } & {\color{blue} 5.5e+11 } & {\color{blue} 4.1e+12 } \\ \hline
s35932 & 35 & 320 & 1728 & 16065 & {\color{blue} 3.4e+10 }  & {\color{blue} 2.6e+11 } & {\color{blue} 1.2e+21 } & {\color{blue} 8.8e+21 } & {\color{blue} 4.1e+31 } & {\color{blue} 3.0e+32 } \\ \hline
s38417 & 28 & 106 & 1636 & 22179 & {\color{blue} 2.7e+08 }  & {\color{blue} 2.0e+09 } & {\color{blue} 7.2e+16 } & {\color{blue} 5.4e+17 } & {\color{blue} 1.9e+25 } & {\color{blue} 1.4e+26 } \\ \hline


s38584 & 11 & 278 & 1452 & 19253 & 2048 & 27394.01 & {\color{blue} 4.2e+06 } & {\color{blue} 3.1e+07 } & {\color{blue} 8.6e+09 } & {\color{blue} 6.4e+10 } \\ \hline


b12 & 5 & 6 & 121 & 1000 &  32 & 55.44 & 1024 & 1934.18  & {\color{blue} 32768 } & {\color{blue} 244449.28 } \\ \hline

b14 & 32 & 54 & 245 & 8567 & {\color{blue} 4.3e+09 }  & {\color{blue} 3.2e+10 } & {\color{blue} 1.8e+19 } & {\color{blue} 1.4e+20 } & {\color{blue} 7.9e+28 } & {\color{blue} 5.9e+29 } \\ \hline
b15 & 36 & 70 & 447 & 6931 & {\color{blue} 6.9e+10 }  & {\color{blue} 5.1e+11 } & {\color{blue} 4.7e+21 } & {\color{blue} 3.5e+22 } & {\color{blue} 3.2e+32 } & {\color{blue} 2.4e+33 } \\ \hline
b18 & 37 & 23 & 20372 & 94249 & {\color{blue} 1.4e+11 }  & {\color{blue} 1.0e+12 } & {\color{blue} 1.9e+22 } & {\color{blue} 1.4e+23 } & {\color{blue} 2.6e+33 } & {\color{blue} 1.9e+34 } \\ \hline
b20 & 32 & 22 & 490 & 17158 & {\color{blue} 4.3e+09 }  & {\color{blue} 3.2e+10 } & {\color{blue} 1.8e+19 } & {\color{blue} 1.4e+20 } & {\color{blue} 7.9e+28 } & {\color{blue} 5.9e+29 } \\


\bottomrule
\end{tabular}}

\end{table}

\noindent\textbf{SAT-Attack Resilience.} Table~\ref{table:SATresults} shows $n_{dip}$ and the runtime resulting from applying the attack on the selected benchmark circuits when $\kappa_s$ ranges from $1$ to $3$, and $\kappa_f$, $\alpha$ and $S$ are fixed to $1$, $0.6$, and $10$, respectively.
With a two-day time-out threshold, four experiments terminated successfully. The results show that the achieved SAT-attack resilience is consistent with~\eqref{eq:sat_resilience_trilock}.
For the rest of the experiments, denoted in blue, we show $n_{dip}$ as computed by \eqref{eq:sat_resilience_trilock} and extrapolate the runtime by conservatively assuming a constant ratio between the runtime and $n_{dip}$ that can be acquired from the finished experiments. 
According to Table~\ref{table:SATresults}, $76.6\%$ of the attack experiments are expected to require more than one year to finish. 

\noindent\textbf{Functional Corruptibility.}
\begin{figure}[t]
    \centering
    \includegraphics[width=\columnwidth]{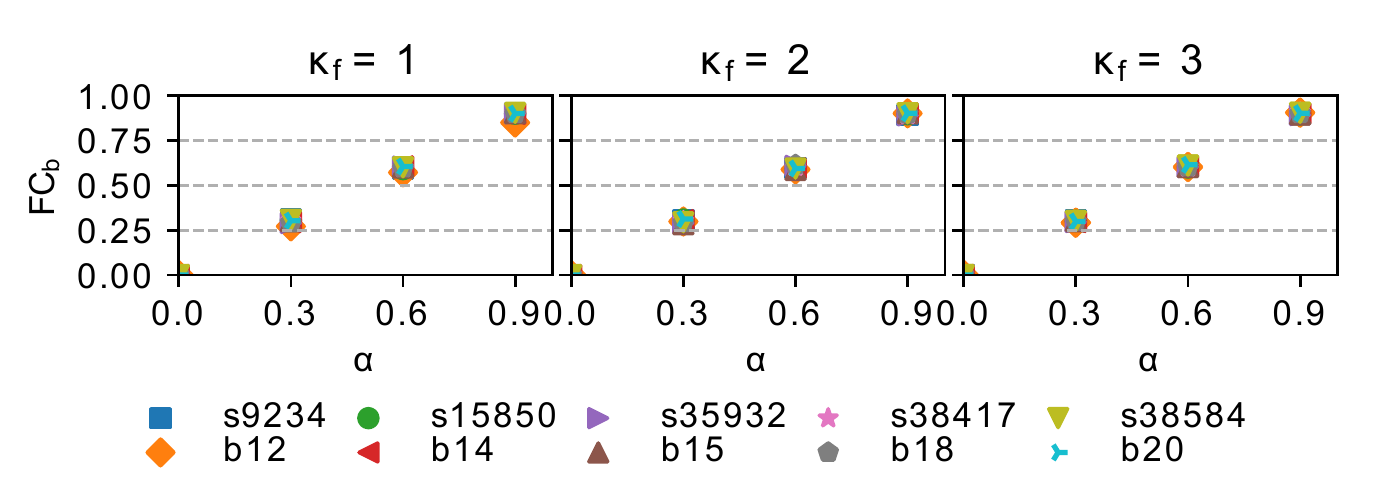}
    \vspace{-9mm}
    \caption{$FC_b$ with different $\alpha$ and $\kappa_f$}
    \label{fig:fc_results}
    \vspace{-4mm}
\end{figure}
Fig.~\ref{fig:fc_results} reports the simulated $FC_b$ for different $\alpha$ and $\kappa_f$. We set $\kappa_s=4$ to achieve high SAT-attack resilience, since $\kappa_s=3$ can already achieve high resilience for most circuits in Table~\ref{table:SATresults}. 
For each locking configuration, we plot the average of the simulated $FC_b$ for $b$ ranging from $\kappa_s$ to $\kappa_s+5$.
Our results show that $FC_b$ is close to its estimate in~\eqref{eq:trilock_max_fc_tune}, with an absolute error within $\pm0.05$, which illustrates $\mathtt{TriLock}$'s ability to configure FC with high SAT-attack resilience. 

\begin{table}[t]
\Large
\centering
\caption{Removal attack resilience of TriLock}
\label{table:SCCresults}
\setlength{\aboverulesep}{0pt}
\setlength{\belowrulesep}{0pt}
\resizebox{0.98\columnwidth}{!}{%

\scriptsize
\begin{tabular}{@{}c|cccc|cccc|cccc}
\toprule
\multirow{2}{*}{Circuit} & \multicolumn{4}{c|}{$S=0$} &  \multicolumn{4}{c|}{$S = 10$} & \multicolumn{4}{c}{$S = 30$} \\ \cline{2-13} 
 
 & $O$ & $E$ & $M$ & $P_M$  & $O$ & $E$ & $M$ & $P_M$  & $O$ & $E$ & $M$  & $P_M$    \\ \hline
s9234 & 72 & 79 & 0 & 0  &12 & 0 & 1 & 95.2 & 0 & 0 & 1 & 100  \\ \hline
s15850 & 203 & 93 & 0 & 0 & 39 & 0 & 1 & 94.0 & 14 & 0 & 1 & 97.9 \\ \hline
s35932 & 18 & 317 & 0 & 0 &  0 & 0 & 1 & 100    & 0 & 0 & 1 & 100  \\ \hline
s38417 & 889 & 198 & 0 & 0 & 36 & 0 & 1 & 97.9 & 20 & 0 & 1 & 98.9 \\ \hline
s38584 & 735 & 79 & 0 & 0 & 30 & 0 & 1 &97.5 & 0 & 0 & 1 & 100\\ \hline
b12 & 19 & 37 & 0 & 0 & 0 & 0 & 1 &100 & 0 & 0  & 1  & 100  \\ \hline
b14 & 57 & 226 & 0  & 0 & 45 & 0 & 1 & 90.4 & 24 & 0 & 1 & 95.1\\ \hline
b15 & 141 & 254 & 0& 0 & 91 & 0 & 1 & 87.1& 61 & 0 & 1 & 91.8 \\ \hline
b18 & 95 & 261 & 0 & 0 &  53 & 0 & 1 & 98.4 & 42 & 0 & 1 & 98.7\\ \hline
b20 & 43 & 226 & 0 & 0 &  31 & 0 & 1 & 95.6 & 10 & 0 & 1 & 98.6 \\ 

\bottomrule
\end{tabular}}
\vspace{-4mm}
\end{table}

\noindent\textbf{Removal Attack Resilience.} 
For each benchmark circuit, we perform state re-encoding with $S=10$ and $S=30$. In addition, we generate a reference design with no state re-encoding, i.e., $S=0$.
Table~\ref{table:SCCresults} shows the results of the SCC algorithm. 
In addition to the number of different types of SCCs, denoted by $O$, $E$, and $M$, we show
the percentage of registers that are in M-SCCs, which is denoted by $P_M$. 
On average, the numbers of O-SCCs and E-SCCs are reduced by $ 71.71 \%$ and $ 100 \%$ when $10$ register pairs are selected for state re-encoding. 
The reduction becomes $ 83.80 \%$ and $ 100 \%$ when $30$ register pairs are selected. 
While state re-encoding may not eliminate the existence of O-SCCs or E-SCCs for most cases in Table~\ref{table:SCCresults}, $P_M$ being close to $100$ indicates that most of the registers are clustered in one M-SCC, which means most of the original and extra registers are densely connected. 


\noindent\textbf{Overhead.} 
We synthesize the locked netlists with $\kappa_f=1$, $\alpha = 0.6$, and $S=10$, which achieve reasonable FC and high removal attack resilience. $\kappa_s$ ranges from $1$ to $5$ to achieve different levels of SAT-attack resilience. 
The overhead of area, delay, and power (ADP) is computed as percentage increase in the area, delay, and power, respectively, incurred by the locking scheme.
We report the ADP overhead 
in Fig.~\ref{fig:overhead},  
showing
that larger circuits tend to exhibit smaller overhead. 
Six out of ten circuits present less than $40\%$ in any of the ADP dimensions. 
In three benchmark circuits, namely, s9234, b14, and b15, the power and area overhead exceed $50\%$ when $\kappa_s>3$. However, as shown in Table~\ref{table:SATresults}, these circuits 
can already achieve reasonably high SAT-attack resilience
with $\kappa_s=2$, where the overhead is less than $40\%$.
In system-on-chip scenarios, it is possible to implement $\mathtt{TriLock}$ only on the sensitive portions of the design, making the 
overhead even smaller. 




\section{Conclusions}\label{sec:conclude}

In this paper, we propose a cost-effective sequential logic locking technique, $\mathtt{TriLock}$, to achieve both high SAT-attack resilience and high functional corruptibility, which circumvents, for the first time, the trade-off between the two security concerns that exists in combinational locking. We also present a state re-encoding technique that can significantly improve the removal attack resilience of $\mathtt{TriLock}$ and, potentially, other sequential locking techniques. Future work includes investigating other attack vectors~\cite{engels2019endof}, e.g., signature analysis on the STG, to further improve the robustness of $\mathtt{TriLock}$.






\bibliography{reference_short} 
\bibliographystyle{ieeetr}

\end{document}